%% file: main.tex
\documentclass[english,american]{IEEEtran}
\usepackage[T1]{fontenc}
\usepackage[utf8]{inputenc}
\usepackage{xcolor}
\usepackage{array}
\usepackage{bm}
\usepackage{multirow}
\usepackage{amsmath}
\usepackage{amssymb}
\usepackage{graphicx}
\PassOptionsToPackage{normalem}{ulem}
\usepackage{ulem}
\usepackage{amsthm}  

\newtheorem{proposition}{Proposition}

\makeatletter

\providecolor{lyxadded}{rgb}{0,0,1}
\providecolor{lyxdeleted}{rgb}{1,0,0}

\DeclareRobustCommand{\lyxsout}[1]{\ifx\\#1\else\sout{#1}\fi}

\usepackage[level=3]{Source/wgroup_message}

\usepackage{color}
\usepackage{psfrag}\usepackage{cite}\usepackage{subfigure}

\usepackage{flushend}

\usepackage[top=0.75in, bottom=1.05in, left=0.625in, right=0.625in]{geometry}

\author{
\IEEEauthorblockN{
	Zheng~Xing\textsuperscript{1,2},
	Mengru Wu\textsuperscript{3},
	Yi Zhang\textsuperscript{4},
	Guanghui Zhang\textsuperscript{5},
	Jun Gao\textsuperscript{1,2},
	Weibing~Zhao\textsuperscript{1,2},\\
	Xuhui Zhang\textsuperscript{1,2},
	Jinke Ren\textsuperscript{2,1},
	and Shuguang Cui\textsuperscript{2,1}.
}

\IEEEauthorblockA{	
		\textsuperscript{1}
		FNii-Shenzhen, CUHKSZ;
	\textsuperscript{2}
	SSE, CUHKSZ;
	\textsuperscript{3}
	Zhejiang University of Technology;\\
	\textsuperscript{4}
	Shenzhen University;
	\textsuperscript{5}
	Shandong University;
}

}

\usepackage[acronym]{glossaries}
\newcommand{\newac}{\newacronym}

\makeglossaries
\input{Source/JCgroup_acronym.tex}

\usepackage{enumitem}

\usepackage{amsmath}   
\usepackage{cuted}     

\newac{ips}{IPS}{indoor positioning system}
\newac{wlan}{WLAN}{wireless local area network}
\newac{ap}{AP}{access point}
\newac{sc}{SC}{subspace clustering}
\newac{uwb}{UWB}{ultrawideband}
\newac{hmm}{HMM}{hidden Markov model}
\newac{ema}{EM}{expectation maximization}
\newac{lbs}{LBS}{location-based service}
\newac{gpc}{GPC}{Gaussian process classification}

\newac{gpca}{GPCA}{generalized PCA}
\newac{msl}{MSL}{multistage learning}
\newac{mppca}{MPPCA}{mixture of probabilistic PCA}
\newac{alc}{ALC}{agglomerative lossy compression}
\newac{ransac}{RANSAC}{random sample consensus}
\newac{fa}{FA}{factorization-based affinity}
\newac{gpcaa}{GPCAA}{gpca-based affinity}
\newac{slbf}{SLBF}{spectral local best-fit flats}
\newac{lsa}{LSA}{local subspace affinity}
\newac{llmc}{LLMC}{locally linear manifold clustering}
\newac{ssc}{SSC}{sparse subspace clustering}
\newac{lrr}{LRR}{low-rank representation}
\newac{scc}{SCC}{spectral curvature clustering}
\newac{mfb}{MFB}{matrix factorization-based}
\newac{mssc}{MSSC}{multi-view low-rank sparse subspace clustering}
\newac{dsc}{DSC}{deep subspace clustering}
\newac{sssc}{SSSC}{Stochastic sparse subspace clustering}
\newac{mdssc}{MDSSC}{}
\newac{nmi}{NMI}{normalized mutual information}
\newac{svm}{SVM}{support vector machine}
\newac{tdoa}{TDOA}{time difference of arrival}
\newac{gmm}{GMM}{Gaussian mixture model}
\newac{pca}{PCA}{Principal Component Analysis}
\newac{dnn}{DNN}{deep neural network}
\newac{mlp}{MLP}{multilayer perceptron}
\newac{admm}{ADMM}{alternating direction method of multipliers}
\newac{gan}{GAN}{generative adversarial networks}
\newac{rsrp}{RSRP}{reference signal received power}
\newac{ss}{SS}{synchronization signal}
\newac{fim}{FIM}{Fisher information matrix}
\newac{crlb}{CRLB}{Cramer-Rao Lower Bound}
\newac{1nn}{1-NN}{1 Nearest-Neighbor}
\newac{dft}{DFT}{discrete Fourier transform}
\newac{evd}{EVD}{Eigen Value Decomposition}
\newac{ssb}{SSB}{secondary synchronization block}
\newac{lstm}{LSTM}{long short-term memory}
\newac{ppp}{PPP}{Poisson Point Process}
\newac{ofdm}{OFDM}{Orthogonal Frequency-Division Multiplexing}
\newac{padp}{PADP}{Power-Angle-Delay Profile}

\newac{imu}{IMU}{Inertial Measurement Unit}

\newac{idft}{IDFT}{Inverse Discrete Fourier Transform}
\newac{lidar}{LiDAR}{Light Detection and Ranging}

\usepackage{babel}
\usepackage{booktabs} 
\usepackage{array}

\usepackage{varwidth}

\addto\captionsamerican{}
\addto\captionsamerican{}

\addto\captionsenglish{}
\addto\captionsenglish{}

\makeatother

\usepackage{babel}
\begin{document}


\title{Annotation-Free Indoor Radio Mapping via Physics-Informed Trajectory Inference}
\maketitle
\begin{abstract}
	Constructing indoor radio maps traditionally requires extensive site surveys with precise user-location labels, making the calibration process costly and time-consuming. Existing calibration-reduction methods either depend on partial location annotations or exploit inertial measurement units (IMUs) to provide relative motion cues; however, IMU-assisted solutions are constrained by hardware availability, device-level access restrictions, and accumulated sensor drift. In this paper, we study a location-label-free indoor radio mapping problem under known access-point deployment geometry and a known walkable spatial domain. We propose a physics-informed trajectory inference framework that uses only Channel State Information (CSI), without relying on user-location labels or IMU measurements. The key idea is to recover the latent spatial coordinates of CSI measurements by exploiting the local spatial continuity of multipath propagation. To this end, we construct a Power-Angle-Delay Profile (PADP) feature distance from MIMO-OFDM CSI and show that, within a local neighborhood and under quasi-static multipath conditions, this distance provides a physically meaningful proxy for small spatial displacements. We then incorporate the PADP-based continuity constraint into a spatially regularized Bayesian inference model for joint trajectory recovery and propagation-parameter estimation. Experiments on a real-world industrial CSI dataset demonstrate that the proposed framework achieves an average localization error of 0.88 m and a relative beam map construction error of 6.68\%, improving upon representative channel-embedding and IMU-assisted baselines.
\end{abstract}


\section{Introduction}
\thispagestyle{empty}

\IEEEPARstart{I}{ndoor} localization underpins a wide range of applications in industrial automation, healthcare, intelligent networking, and emerging 6G sensing systems~\cite{zafari2019survey}. Among different sensing modalities, including camera, LiDAR, visible light, and inertial sensing, CSI-based localization~\cite{saeed2019state} has attracted considerable attention because it can exploit existing wireless infrastructure without requiring dedicated visual or ranging hardware.

Fingerprint-based localization is one of the most widely adopted paradigms for indoor positioning. It constructs a radio map by associating wireless measurements with accurately surveyed physical coordinates and then localizes users by matching online CSI observations against the calibrated map. Although this paradigm can achieve high positioning accuracy, it fundamentally relies on labor-intensive site surveys and precise user-location labels, making radio map construction expensive and difficult to scale~\cite{zhu2020indoor}. To reduce calibration costs, prior studies have explored generative models, transfer learning, interpolation, and channel charting techniques~\cite{TanPal:J25,StaYam:J23}. Nevertheless, many of these approaches still require partial location annotations, anchor points, or post-hoc geometric alignment, and thus do not fully remove the dependence on labeled fingerprints.

Another line of work attempts to reduce or avoid explicit location labels by incorporating IMUs, which provide relative motion information such as acceleration, heading, and step displacement~\cite{li2021wifi,zhao2020graphips}. These methods can be effective when reliable motion-sensor data are available. However, their practical deployment is limited by hardware heterogeneity, operating-system permission constraints, privacy policies, and accumulated drift over long trajectories~\cite{si2024environment,si2025unsupervised}. Therefore, an important open question is whether indoor radio maps can be constructed from CSI sequences without user-location labels and without auxiliary IMU measurements.

In this paper, we investigate this problem under a practical but explicitly stated deployment setting: the access-point (AP) locations, array orientations, and valid walkable spatial domain are assumed to be known from installation records or floor-plan information, while no CSI sample is annotated with the user's physical coordinate. Under this setting, the main challenge is to infer the latent trajectory associated with a temporal CSI sequence and use the recovered trajectory to support radio map construction. This problem is highly non-trivial in complex indoor environments, where severe multipath propagation, asynchronous measurements, and non-line-of-sight components make the relationship between CSI variations and physical displacement difficult to model.

Our key observation is that local angular-delay variations of CSI retain useful geometric information despite severe multipath fading. Under small displacements in a quasi-static environment, dominant multipath components induce structured CSI changes. We therefore construct a PADP feature distance from MIMO-OFDM CSI and use it as a physics-informed spatial continuity cue, where nearby locations tend to exhibit smaller PADP discrepancies.
Building on this cue, we propose a spatially regularized Bayesian trajectory inference framework. It integrates RSS, directional information, mobility constraints, and PADP-based continuity regularization to jointly recover the latent user trajectory and estimate propagation-related parameters. The resulting problem is solved by alternating parameter updates and regularized hidden-state decoding over the feasible spatial graph.

The main contributions are summarized as follows:
\begin{itemize}
	\item \textbf{Location-label-free formulation:} We study indoor radio mapping without user-location labels or IMU measurements, while explicitly assuming known AP geometry and a valid walkable domain.
	
	\item \textbf{PADP continuity cue:} We construct a PADP-based CSI distance that captures local angular-delay variations and serves as a physics-informed proxy for small spatial displacements.
	
	\item \textbf{Regularized Bayesian inference:} We develop a Bayesian trajectory inference framework that combines RSS, angle information, mobility constraints, and PADP continuity to refine the latent trajectory and propagation parameters.
	
	\item \textbf{Real-world validation:} Experiments on an industrial CSI dataset show that the proposed method achieves a localization error of \(0.88\,\mathrm{m}\) and a relative beam map error of \(6.68\%\), outperforming representative baselines.
\end{itemize}
\section{System Model}
\subsection{Propagation Modeling in MIMO-OFDM Systems}
\label{sec:propog-model}

We consider an indoor wireless environment served by \(Q\) APs located at known coordinates \(\mathbf{o}_q\in\mathbb{R}^2\), \(q\in\{1,\ldots,Q\}\). Their array orientations and the valid walkable region are assumed to be available from deployment records or floor-plan information. These quantities are deployment-side priors, not user-location annotations; no CSI sample is associated with a ground-truth user coordinate during training or inference.

A single-antenna mobile terminal moves in the target region and periodically exchanges OFDM signals with the APs. We present the model in the downlink form, while the same formulation applies to uplink channel-sounding measurements under channel reciprocity and array calibration.

Each AP has \(N_{\mathrm{ant}}\) antenna elements. To cover both Uniform Linear Array (ULA) and Uniform Planar Array (UPA) deployments, we adopt a general array-response model. Let \(\mathbf{r}_{q,n}\in\mathbb{R}^d\) be the local coordinate of the \(n\)-th antenna element of AP \(q\), and let \(\mathbf{R}_q\) map the local array frame to the global frame. For a path with direction parameter \(\boldsymbol{\vartheta}\), the steering vector is
\begin{align}
	\label{eq:general_steering}
	\mathbf{a}_q(\boldsymbol{\vartheta})
	=
	\left[
	e^{-j\frac{2\pi}{\lambda}\mathbf{u}^{\mathrm{T}}(\boldsymbol{\vartheta})\mathbf{R}_q\mathbf{r}_{q,1}},
	\ldots,
	e^{-j\frac{2\pi}{\lambda}\mathbf{u}^{\mathrm{T}}(\boldsymbol{\vartheta})\mathbf{R}_q\mathbf{r}_{q,N_{\mathrm{ant}}}}
	\right]^{\mathrm{T}}
\end{align}
where \(\lambda=c/f_c\), and \(\mathbf{u}(\boldsymbol{\vartheta})\) is the unit propagation-direction vector. For a two-dimensional azimuth model, \(\boldsymbol{\vartheta}=\theta\) and \(\mathbf{u}(\theta)=[\cos\theta,\sin\theta]^{\mathrm{T}}\). For a ULA with spacing \(\Delta\), this reduces to \(\mathbf{a}_q(\varphi)= [1,e^{-j\frac{2\pi}{\lambda}\Delta\sin\varphi},\ldots,e^{-j\frac{2\pi}{\lambda}(N_{\mathrm{ant}}-1)\Delta\sin\varphi}]^{\mathrm{T}}\), where \(\varphi\) is the relative broadside angle.

CSI is sampled every \(\delta\) seconds over \(t\in\{1,\ldots,T\}\), and the unknown user position is \(\mathbf{x}_t\in\mathbb{R}^2\). For AP \(q\), the \(\ell\)-th path has direction \(\boldsymbol{\vartheta}_{t,q}^{(\ell)}\), delay \(\tau_{t,q}^{(\ell)}\), and complex gain \(\kappa_{t,q}^{(\ell)}\). In the azimuth-only case, the absolute and relative angles satisfy \(\theta_{t,q}^{(\ell)}=(\varphi_{t,q}^{(\ell)}+\phi_q)\bmod 2\pi\), where \(\phi_q\) is the known AP orientation.

The OFDM system uses bandwidth \(B\), \(N_{\mathrm{sub}}\) subcarriers, and spacing \(\Delta f=B/N_{\mathrm{sub}}\). The array-domain channel vector of AP \(q\) on subcarrier \(m\in\{0,\ldots,N_{\mathrm{sub}}-1\}\) is
\begin{align}
	\label{eq:channel}
	\mathbf{h}_{t,q}^{(m)}
	=
	\sum_{\ell=1}^{L_{t,q}}
	\kappa_{t,q}^{(\ell)}
	e^{-j2\pi m\Delta f \tau_{t,q}^{(\ell)}}
	\mathbf{a}_q\!\left(\boldsymbol{\vartheta}_{t,q}^{(\ell)}\right)
	+
	\mathbf{n}_{t,q}^{(m)},
\end{align}
where \(L_{t,q}\) is the number of resolvable paths and \(\mathbf{n}_{t,q}^{(m)}\) denotes measurement noise and modelling error. In asynchronous channel sounding, a common timing or phase offset may appear across subcarriers; the proposed PADP feature therefore relies mainly on relative angular-delay power structure rather than absolute global phase.

Stacking all subcarriers gives \(\mathbf{H}_{t,q}=[\mathbf{h}_{t,q}^{(0)},\ldots,\mathbf{h}_{t,q}^{(N_{\mathrm{sub}}-1)}]\in\mathbb{C}^{N_{\mathrm{ant}}\times N_{\mathrm{sub}}}\), and the multi-AP observation is \(\mathbf{Y}_t=\{\mathbf{H}_{t,q}\}_{q=1}^{Q}\). The goal is to infer the latent trajectory \(\mathbf{X}_T=(\mathbf{x}_1,\ldots,\mathbf{x}_T)\) and construct the radio map from \(\{\mathbf{Y}_t\}_{t=1}^{T}\), without using ground-truth user-location labels or IMU measurements.
\subsection{Mobility Model}

To make trajectory inference tractable, we discretize the known walkable region into a finite spatial graph. Let \(\mathcal{V}=\{\mathbf{p}_1,\mathbf{p}_2,\ldots,\mathbf{p}_N\}\subset\mathbb{R}^2\) denote the set of valid spatial nodes obtained from the floor-plan-constrained domain. The hidden state at time \(t\) is represented by a graph node, i.e., \(\mathbf{x}_t\in\mathcal{V}\).

The graph edges encode feasible movements between adjacent nodes. Let \(v_{\max}\) be the maximum plausible walking speed and \(\delta\) be the CSI sampling interval. The maximum displacement within one interval is \(D_{\mathrm{m}}=v_{\max}\delta\). The feasible edge set is defined as
$\mathcal{E}
=
\left\{
(\mathbf{p}_i,\mathbf{p}_j)
\middle|
\mathbf{p}_i,\mathbf{p}_j\in\mathcal{V},
\|\mathbf{p}_i-\mathbf{p}_j\|_2\leq D_{\mathrm{m}},
\operatorname{seg}(\mathbf{p}_i,\mathbf{p}_j)\subseteq\mathcal{R}
\right\}$,
where \(\mathcal{R}\) is the valid walkable region and \(\operatorname{seg}(\mathbf{p}_i,\mathbf{p}_j)\) denotes the line segment connecting \(\mathbf{p}_i\) and \(\mathbf{p}_j\). This condition excludes transitions crossing walls, obstacles, or non-walkable regions.

The user mobility is modeled as a first-order Markov process over the graph, i.e., \(\mathbb{P}(\mathbf{x}_t=\mathbf{p}_j\mid\mathbf{x}_{t-1}=\mathbf{p}_i)=P_{ij}\). We use a truncated Gaussian distance kernel:
\begin{align}
	\label{eq:transition_probability}
	P_{ij}
	=
	\begin{cases}
		\dfrac{1}{Z_i}
		\exp\!\left(
		-\dfrac{\|\mathbf{p}_i-\mathbf{p}_j\|_2^2}{2\sigma_{\mathrm{m}}^2}
		\right),
		&
		\text{if }(\mathbf{p}_i,\mathbf{p}_j)\in\mathcal{E},
		\\[2ex]
		0,
		&
		\text{otherwise},
	\end{cases}
\end{align}
where \(\sigma_{\mathrm{m}}^2\) controls mobility uncertainty, and \(Z_i=\sum_{\{\mathbf{p}_k\,|\,(\mathbf{p}_i,\mathbf{p}_k)\in\mathcal{E}\}}\exp(-\|\mathbf{p}_i-\mathbf{p}_k\|_2^2/(2\sigma_{\mathrm{m}}^2))\) ensures \(\sum_{j=1}^{N}P_{ij}=1\). This mobility model provides a weak kinematic prior and requires only the sampling interval, a plausible speed bound, and the feasible spatial domain, rather than labeled user positions.

\section{Methodology}
\label{subsec:PDFmodel}
\subsection{CSI-Based Spatial Continuity Model}

We construct a physics-informed CSI feature distance to capture local angular-delay variations of multipath propagation. The motivation is that small physical displacements induce structured changes in the phase, delay, and angular responses of dominant paths.

For AP \(q\), let \(\mathbf{A}_q=[\mathbf{a}_q(\boldsymbol{\vartheta}_1),\ldots,\mathbf{a}_q(\boldsymbol{\vartheta}_{N_{\mathrm{a}}})]\in\mathbb{C}^{N_{\mathrm{ant}}\times N_{\mathrm{a}}}\) denote the angular dictionary constructed from \eqref{eq:general_steering}, where \(N_{\mathrm{a}}\) is the number of angular grid points. Let \(\mathbf{F}\in\mathbb{C}^{N_{\mathrm{sub}}\times N_{\mathrm{sub}}}\) be the normalized DFT matrix. To suppress global power variations, we normalize the CSI as \(\bar{\mathbf{H}}_{t,q}=\mathbf{H}_{t,q}/(\|\mathbf{H}_{t,q}\|_{\mathrm{F}}+\epsilon_{\mathrm{H}})\), where \(\epsilon_{\mathrm{H}}>0\) is a small stabilizing constant. The PADP representation is then defined as
\begin{align}
	\label{eq:padp_feature}
	\mathbf{P}_{t,q}
	=
	\left|
	\mathbf{A}_q^{\mathrm{H}}
	\bar{\mathbf{H}}_{t,q}
	\mathbf{F}^{\mathrm{H}}
	\right|,
\end{align}
where \(|\cdot|\) is applied element-wise. For two CSI observations \(\mathbf{H}_{i,q}\) and \(\mathbf{H}_{j,q}\), their PADP distance is
\begin{align}
	\label{eq:padp_distance}
	g_{i,j,q}
	=
	D(\mathbf{H}_{i,q},\mathbf{H}_{j,q})
	=
	\frac{1}{\sqrt{N_{\mathrm{a}}N_{\mathrm{sub}}}}
	\|\mathbf{P}_{i,q}-\mathbf{P}_{j,q}\|_{\mathrm{F}},
\end{align}
where the normalization reduces sensitivity to the angular-delay grid size.

\begin{proposition}[Local PADP Continuity]
	\label{prop:local-continuity}
	For two nearby locations \(\mathbf{x}_i\) and \(\mathbf{x}_j\) with \(d_{i,j}=\|\mathbf{x}_i-\mathbf{x}_j\|_2<\rho\), suppose that the dominant scatterers remain quasi-static, path parameters vary smoothly, no dominant path birth or death occurs, and the angular-delay dictionary has sufficient resolution. Then there exist constants \(c_{1,q}\geq0\), \(c_{2,q}>0\), and \(b_q\geq0\) such that
$$c_{1,q}d_{i,j}-b_q
\leq
\mathbb{E}
\!\left[
D(\mathbf{H}_{i,q},\mathbf{H}_{j,q})
\mid
\mathbf{x}_i,\mathbf{x}_j
\right]
\leq
c_{2,q}d_{i,j}+b_q.$$
\end{proposition}

\begin{proof}[Proof Sketch]
	Within a small neighborhood, dominant multipath components are locally stable, and their delay and angular perturbations admit first-order expansions with respect to \(\mathbf{x}_i-\mathbf{x}_j\). Since \(\mathbf{P}_{t,q}\) is obtained by bounded linear projections followed by element-wise magnitude extraction, the PADP variation is locally Lipschitz with physical displacement. Under a non-degenerate angular-delay response, a local lower bound also holds up to the modelling error \(b_q\).
\end{proof}

Proposition~\ref{prop:local-continuity} does not claim a globally exact linear law. It only motivates the local surrogate \(g_{i,j,q}=\gamma_{1,q}d_{i,j}+\nu_{i,j,q}\), where \(\mathbb{E}[\nu_{i,j,q}]=0\) and \(\mathrm{Var}(\nu_{i,j,q})=\gamma_{2,q}d_{i,j}+\sigma_{0,q}^{2}\). Here, \(\gamma_{1,q}>0\), \(\gamma_{2,q}>0\), and \(\sigma_{0,q}^{2}>0\) are AP-dependent continuity parameters. This surrogate is used as the pairwise regularization model for trajectory inference.

\subsection{CSI-Based Power and Direction Model}

In addition to the PADP distance, the CSI matrix also provides macroscopic power and directional information. The effective received signal strength (RSS) associated with AP \(q\) at time \(t\) is computed as
\begin{align}
	\label{eq:rss}
	r_{t,q}
	=
	10\log_{10}
	\left(
	\|\mathbf{H}_{t,q}\|_{\mathrm{F}}^{2}
	+
	\epsilon_{\mathrm{H}}
	\right).
\end{align}
Given a candidate user position \(\mathbf{x}_t\), the RSS is modeled by a log-distance path-loss model:
\begin{align}
	\label{eq:rss_likelihood}
	r_{t,q}
	=
	\beta_q
	-
	\alpha_q
	\log_{10}
	\left(
	\|\mathbf{x}_t-\mathbf{o}_q\|_2+d_0
	\right)
	+
	\xi_{t,q},
\end{align}
where $\xi_{t,q}
\sim
\mathcal{N}(0,\sigma_{s,q}^{2})$, \(\beta_q\) is the reference power parameter, \(\alpha_q\) is the path-loss coefficient, \(d_0>0\) avoids the singularity at zero distance, and \(\sigma_{s,q}^{2}\) captures shadowing and hardware-induced power fluctuations.

To extract directional information, we compute the sample spatial covariance matrix
$\mathbf{V}_{t,q}
=
\frac{1}{N_{\mathrm{sub}}}
\mathbf{H}_{t,q}\mathbf{H}_{t,q}^{\mathrm{H}}$.
Let \(\mathbf{U}_{t,q}\) denote the estimated noise subspace obtained from the eigendecomposition of \(\mathbf{V}_{t,q}\). The dominant direction is estimated by a MUSIC-type search:
\begin{align}
	\label{eq:music_general}
	\hat{\boldsymbol{\vartheta}}_{t,q}
	=
	\underset{\boldsymbol{\vartheta}\in\mathcal{S}_q}{\arg\max}
	\,
	\frac{1}
	{
		\mathbf{a}_q^{\mathrm{H}}(\boldsymbol{\vartheta})
		\mathbf{U}_{t,q}\mathbf{U}_{t,q}^{\mathrm{H}}
		\mathbf{a}_q(\boldsymbol{\vartheta})
	},
\end{align}
where \(\mathcal{S}_q\) denotes the angular search region of AP \(q\). In the two-dimensional azimuth-only case, we write the estimated bearing as \(\hat{\theta}_{t,q}\).

The dominant bearing is treated as an effective macroscopic directional cue. Since the strongest path may correspond to either a line-of-sight component or a stable specular reflection, the following angular likelihood should be interpreted as a robust approximation rather than a perfect geometric measurement:
\begin{align}
	\label{eq:angle_likelihood}
	p(\hat{\theta}_{t,q}\mid\mathbf{x}_t)
	\propto
	\exp
	\left(
	-
	\frac{
		\operatorname{wrap}^{2}
		\left(
		\hat{\theta}_{t,q}
		-
		\phi(\mathbf{x}_t,\mathbf{o}_q)
		\right)
	}
	{2\sigma_{\theta}^{2}}
	\right),
\end{align}
where \(\phi(\mathbf{x}_t,\mathbf{o}_q)\) denotes the geometric bearing from AP \(q\) to \(\mathbf{x}_t\), \(\sigma_{\theta}^{2}\) is the angular uncertainty, and \(\operatorname{wrap}(\cdot)\) maps an angle difference to \((-\pi,\pi]\).

\subsection{Spatially Regularized Bayesian Formulation}

Let \(\mathcal{X}_T=(\mathbf{x}_1,\ldots,\mathbf{x}_T)\) denote the hidden trajectory and \(\mathcal{Y}_T=\{\mathbf{Y}_t\}_{t=1}^{T}\) denote the CSI sequence. For a candidate state \(\mathbf{x}_t\), the local emission log-likelihood is
\begin{align}
	\label{eq:emission_score}
	\ell_t(\mathbf{x}_t)
	=
	\sum_{q=1}^{Q}
	\left[
	\log p(r_{t,q}\mid\mathbf{x}_t)
	+
	\log p(\hat{\theta}_{t,q}\mid\mathbf{x}_t)
	\right],
\end{align}
where the RSS and angular likelihoods are given by \eqref{eq:rss_likelihood} and \eqref{eq:angle_likelihood}, respectively.

To make the inference compatible with first-order hidden-state decoding, we impose PADP continuity on adjacent temporal pairs. For \(t\geq2\), let \(d_t(\mathbf{x}_{t-1},\mathbf{x}_t)=\|\mathbf{x}_t-\mathbf{x}_{t-1}\|_2\). Following the AP-dependent surrogate implied by Proposition~\ref{prop:local-continuity}, the PADP pairwise energy for AP \(q\) is
\begin{align}
	\label{eq:padp_pairwise_energy}
	\psi_{t,q}(\mathbf{x}_{t-1},\mathbf{x}_t)
	&=
	\frac{
		\left(
		g_{t,t-1,q}
		-
		\gamma_{1,q}d_t(\mathbf{x}_{t-1},\mathbf{x}_t)
		\right)^2
	}
	{
		2\left(
		\gamma_{2,q}d_t(\mathbf{x}_{t-1},\mathbf{x}_t)
		+
		\sigma_{0,q}^{2}
		\right)
	}
	\\
	&+
	\frac{1}{2}
	\log
	\left(
	\gamma_{2,q}d_t(\mathbf{x}_{t-1},\mathbf{x}_t)
	+
	\sigma_{0,q}^{2}
	\right).\nonumber
\end{align}
This energy penalizes transitions whose physical displacement is inconsistent with the observed PADP variation. Here, \(\gamma_{1,q}\), \(\gamma_{2,q}\), and \(\sigma_{0,q}^{2}\) are AP-dependent continuity parameters.

Let
$
\bm{\Theta}_{\mathrm{p}}
=
\{\beta_q,\alpha_q,\sigma_{s,q}^{2},\gamma_{1,q},\gamma_{2,q}\}_{q=1}^{Q}
\cup
\{\sigma_{\theta}^{2}\}$.
The spatially regularized objective is
\begin{align}
	\label{eq:P0_revised}
	\underset{\mathcal{X}_T,\bm{\Theta}_{\mathrm{p}}}{\mathrm{maximize}}
	\quad
	\mathcal{J}(\mathcal{X}_T,\bm{\Theta}_{\mathrm{p}})
	=
	&
	\sum_{t=1}^{T}
	\ell_t(\mathbf{x}_t)
	+
	\sum_{t=2}^{T}
	\log
	\mathbb{P}(\mathbf{x}_t\mid\mathbf{x}_{t-1})
	\nonumber\\
	&
	-
	\eta
	\sum_{t=2}^{T}
	\sum_{q=1}^{Q}
	\psi_{t,q}(\mathbf{x}_{t-1},\mathbf{x}_t),
\end{align}
where \(\eta\geq0\) balances measurement fitting and PADP-based spatial continuity.

\textit{Remark:}
A small \(\eta\) may underuse the CSI-derived geometric cue, whereas an excessively large \(\eta\) may over-regularize the trajectory. Thus, we do not assume that the regularizer is universally beneficial for all \(\eta\). In practice, \(\eta\) is selected by an annotation-free criterion based on the normalized objective in \eqref{eq:P0_revised} and the PADP residual consistency.

\subsection{Spatially Regularized Trajectory Inference Algorithm}
\label{subsec:algorithm}

The trajectory \(\mathcal{X}_T\) and parameters \(\bm{\Theta}_{\mathrm{p}}\) are coupled in \eqref{eq:P0_revised}. We therefore use alternating optimization: given the current trajectory, update the parameters; given the parameters, refine the trajectory by regularized hidden-state decoding.

\subsubsection{Initialization}

We initialize the trajectory by solving the unregularized problem with \(\eta=0\):
\begin{align}
	\hat{\mathcal{X}}_T^{(0)}
	=
	\arg\max_{\mathcal{X}_T}
	\left[
	\sum_{t=1}^{T}\ell_t(\mathbf{x}_t)
	+
	\sum_{t=2}^{T}
	\log
	\mathbb{P}(\mathbf{x}_t\mid\mathbf{x}_{t-1})
	\right].
\end{align}
This gives a coarse trajectory based on RSS, directional likelihood, and mobility constraints.

\subsubsection{Parameter Estimation}

Given \(\hat{\mathcal{X}}_T^{(k-1)}\), we update \(\bm{\Theta}_{\mathrm{p}}\).

\paragraph{Update of Path-Loss Parameters}
For AP \(q\), define \(u_{t,q}^{(k-1)}=\log_{10}(\|\hat{\mathbf{x}}_{t}^{(k-1)}-\mathbf{o}_q\|_2+d_0)\). The RSS model becomes \(r_{t,q}=\beta_q-\alpha_q u_{t,q}^{(k-1)}+\xi_{t,q}\). Let \(\mathbf{s}_q=[r_{1,q},\ldots,r_{T,q}]^{\mathrm{T}}\), and let the \(t\)-th row of \(\mathbf{B}_q\) be \([1,-u_{t,q}^{(k-1)}]\). Then, $[\hat{\beta}_q,\hat{\alpha}_q]^{\text{T}}
=
(\mathbf{B}_q^{\mathrm{T}}\mathbf{B}_q)^{-1}
\mathbf{B}_q^{\mathrm{T}}\mathbf{s}_q$.
The RSS variance is updated by
\begin{align}
	\label{eq:rss_var_update}
	\hat{\sigma}_{s,q}^{2}
	=
	\frac{1}{T}
	\sum_{t=1}^{T}
	\left(
	r_{t,q}
	-
	\hat{\beta}_q
	+
	\hat{\alpha}_q u_{t,q}^{(k-1)}
	\right)^2 .
\end{align}

\paragraph{Update of Angular Noise Variance}
The angular variance is updated as
\begin{align}
	\label{eq:angle_var_update}
	\hat{\sigma}_{\theta}^{2}
	=
	\frac{1}{TQ}
	\sum_{t=1}^{T}
	\sum_{q=1}^{Q}
	\operatorname{wrap}^{2}
	\left(
	\hat{\theta}_{t,q}
	-
	\phi(\hat{\mathbf{x}}_{t}^{(k-1)},\mathbf{o}_q)
	\right).
\end{align}

\paragraph{Update of PADP Continuity Parameters}
For \(t\geq2\), define \(\hat{d}_{t}^{(k-1)}=\|\hat{\mathbf{x}}_{t}^{(k-1)}-\hat{\mathbf{x}}_{t-1}^{(k-1)}\|_2\). For each AP \(q\), the slope parameter is updated as
\begin{align}
	\label{eq:gamma1_update}
	\hat{\gamma}_{1,q}
	=
	\frac{
		\sum_{t=2}^{T}
		g_{t,t-1,q}
		\hat{d}_{t}^{(k-1)}
	}
	{
		\sum_{t=2}^{T}
		(\hat{d}_{t}^{(k-1)})^2
		+
		\epsilon_{\mathrm{d}}
	},
\end{align}
where \(\epsilon_{\mathrm{d}}>0\) is a stabilizing constant. The variance coefficient is updated by
\begin{align}
	\label{eq:gamma2_update}
	\hat{\gamma}_{2,q}
	=
	\max
	\left\{
	\gamma_{\min},
	\frac{
		\sum_{t=2}^{T}
		\left(
		g_{t,t-1,q}
		-
		\hat{\gamma}_{1,q}\hat{d}_{t}^{(k-1)}
		\right)^2
	}
	{
		\sum_{t=2}^{T}
		\hat{d}_{t}^{(k-1)}
		+
		\epsilon_{\mathrm{d}}
	}
	\right\},
\end{align}
where \(\gamma_{\min}>0\) prevents degenerate variance estimates. The floor \(\sigma_{0,q}^{2}\) is fixed as a small positive constant or estimated from very small-displacement pairs when available.

\subsubsection{Trajectory Optimization}

Given \(\hat{\bm{\Theta}}_{\mathrm{p}}^{(k)}\), the trajectory is refined by maximizing \eqref{eq:P0_revised} over the spatial graph. Since the objective contains unary emission terms and first-order pairwise terms, it admits a Viterbi-type recursion. Let \(\Omega_t(\mathbf{x})\) be the maximum accumulated score of paths ending at \(\mathbf{x}\in\mathcal{V}\) at time \(t\), with initialization \(\Omega_1(\mathbf{x})=\ell_1(\mathbf{x})\). For \(t\geq2\),
\begin{align}
	\Omega_t(\mathbf{x})
	=
	\ell_t(\mathbf{x})\nonumber
	&+
	\max_{\mathbf{x}'\in\mathcal{V}:\,(\mathbf{x}',\mathbf{x})\in\mathcal{E}}
	\{
	\Omega_{t-1}(\mathbf{x}')
	+
	\log\mathbb{P}(\mathbf{x}\mid\mathbf{x}')
	\\-
	&\eta
	\sum_{q=1}^{Q}
	\psi_{t,q}(\mathbf{x}',\mathbf{x})	\label{eq:viterbi_recursion_revised}
	\}.
\end{align}
The optimal trajectory is recovered by standard backtracking.

For efficiency, we retain a candidate set \(\mathcal{C}_t\subseteq\mathcal{V}\) with the largest emission scores at each time step and run \eqref{eq:viterbi_recursion_revised} over \(\mathcal{C}_{t-1}\times\mathcal{C}_t\). If \(n_{\max}=\max_t|\mathcal{C}_t|\) and \(\bar{\rho}(D_{\mathrm{m}})\) is the maximum number of feasible predecessor states, the decoding complexity is \(\mathcal{O}(Tn_{\max}\bar{\rho}(D_{\mathrm{m}})Q)\). The procedure stops when the relative objective change is below a threshold or the maximum number of iterations is reached.

\subsection{Radio Map Construction From the Recovered Trajectory}

After obtaining \(\hat{\mathcal{X}}_T=(\hat{\mathbf{x}}_1,\ldots,\hat{\mathbf{x}}_T)\), we associate each CSI observation \(\mathbf{H}_{t,q}\) with its inferred coordinate \(\hat{\mathbf{x}}_t\), yielding a location-labeled radio-map surrogate without ground-truth user-location annotations.

For any spatial node \(\mathbf{p}\in\mathcal{V}\), we define the kernel weight as \(\omega_t(\mathbf{p})=\exp(-\|\mathbf{p}-\hat{\mathbf{x}}_t\|_2^2/(2h^2))\), where \(h>0\) is the smoothing bandwidth. The RSS map of AP \(q\) is estimated by \(\hat{r}_q(\mathbf{p})={\sum_{t=1}^{T}\omega_t(\mathbf{p})r_{t,q}}/({\sum_{t=1}^{T}\omega_t(\mathbf{p})+\epsilon_{\mathrm{w}}})\), and the PADP-domain map is estimated by \(\hat{\mathbf{P}}_q(\mathbf{p})={\sum_{t=1}^{T}\omega_t(\mathbf{p})\mathbf{P}_{t,q}}/({\sum_{t=1}^{T}\omega_t(\mathbf{p})+\epsilon_{\mathrm{w}}})\), where \(\epsilon_{\mathrm{w}}>0\) is a stabilizer.

If a beam codebook \(\mathcal{W}_q=\{\mathbf{w}_{q,1},\ldots,\mathbf{w}_{q,B_{\mathrm{c}}}\}\) is available, we first construct the local channel map from phase-normalized CSI as \(\hat{\mathbf{H}}_q(\mathbf{p})={\sum_{t=1}^{T}\omega_t(\mathbf{p})\tilde{\mathbf{H}}_{t,q}}/({\sum_{t=1}^{T}\omega_t(\mathbf{p})+\epsilon_{\mathrm{w}}})\). The predicted gain of beam \(b\) is then
\begin{align}
	\label{eq:beam_gain_map}
	\hat{e}_{q,b}(\mathbf{p})
	=
	\frac{1}{N_{\mathrm{sub}}}
	\sum_{m=0}^{N_{\mathrm{sub}}-1}
	\left|
	\mathbf{w}_{q,b}^{\mathrm{H}}
	\hat{\mathbf{h}}_{q}^{(m)}(\mathbf{p})
	\right|^2,
\end{align}
where \(\hat{\mathbf{h}}_{q}^{(m)}(\mathbf{p})\) is the \(m\)-th subcarrier column of \(\hat{\mathbf{H}}_q(\mathbf{p})\). The best beam is selected as \(\hat{b}_{q}^{\star}(\mathbf{p})=\arg\max_{b\in\{1,\ldots,B_{\mathrm{c}}\}}\hat{e}_{q,b}(\mathbf{p})\). Thus, the recovered trajectory provides spatial anchors for constructing RSS, PADP, channel, and beam-domain radio maps.

\section{Numerical Experiments}
\label{sec:Experiments}

\begin{figure}[t]
	\centering
	\includegraphics[width=1\columnwidth]{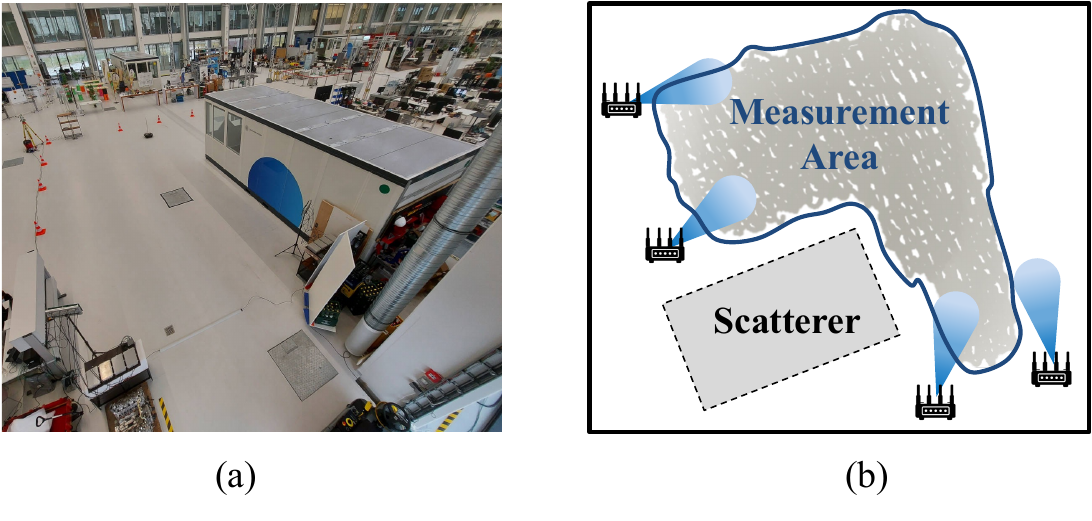}
	\vspace{-0.3in}
	\caption{Measured industrial environment used for evaluation: (a) top-view layout of the facility and (b) scatter plot of the ground-truth positions used only for performance evaluation. The antenna arrays are represented by black icons, while the blue sectors indicate their effective fields of view.}
	\label{fig:DatasetIEnv}
	\vspace{-0.2in}
\end{figure}
\subsection{Experimental Setup}
\label{subsec:exp-setup}

We evaluate the proposed framework on the \emph{dichasus-cf0x} dataset~\cite{dataset-dichasus-cf0x}, collected in a real industrial hall with an L-shaped measurement area of approximately \(14\,\mathrm{m}\times14\,\mathrm{m}\), as shown in Fig.~\ref{fig:DatasetIEnv}(a). The deployment consists of \(Q=4\) distributed UPAs, each with \(N_{\mathrm{ant}}=N_{\mathrm{row}}\times N_{\mathrm{col}}=2\times4\) half-wavelength-spaced antenna elements. The channel sounder records \(N_{\mathrm{sub}}=1024\) OFDM subcarriers over a \(50\,\mathrm{MHz}\) bandwidth centered at \(f_c=1.272\,\mathrm{GHz}\), and the subcarriers are divided into \(U=4\) groups for efficient processing.

At time \(t\), the multi-AP CSI observation is \(\mathbf{Y}_t=\{\mathbf{H}_{t,q}\}_{q=1}^{Q}\). The distributed arrays are mutually synchronized, whereas the mobile terminal and receiving arrays are not fully synchronized, introducing an unknown time-of-transmission offset. Since the PADP feature in \eqref{eq:padp_feature} uses angular-delay power profiles, it mainly exploits relative multipath structure and is less dependent on absolute global phase.

The mobile terminal follows multiple trajectories in the measurement area. A high-precision tachymeter provides millimeter-level ground-truth coordinates \(\mathbf{x}_t\), which are used only for evaluation and visualization, not for trajectory inference, parameter estimation, or hyperparameter selection. The dataset is denoted by \(\mathcal{S}=\{(\{\mathbf{H}_{t,q}\}_{q=1}^{Q},\mathbf{x}_t)\}_{t=1}^{T}\). Since the antenna height is fixed, trajectory inference is evaluated on the two-dimensional floor plane.

\subsection{Metrics and Comparative Baselines}
\label{subsec:methods}

We evaluate trajectory recovery and radio map construction using four metrics. The average localization error is \(E_{\mathrm{loc}}=\frac{1}{T}\sum_{t=1}^{T}\|\hat{\mathbf{x}}_t-\mathbf{x}_t\|_2\), where \(\hat{\mathbf{x}}_t\) is the inferred coordinate and \(\mathbf{x}_t\) is used only for evaluation. The relative beam estimation error is \(E_{\mathrm{Beam}}=\frac{1}{TQ}\sum_{t=1}^{T}\sum_{q=1}^{Q}\left|(e_{t,q}-\hat{e}_{t,q})/(e_{t,q}+\epsilon_{\mathrm{e}})\right|\times100\%\), where \(e_{t,q}\) and \(\hat{e}_{t,q}\) are the reference and estimated beam gains. The channel reconstruction error is \(E_{\mathrm{RMSE}}=\frac{1}{TQ}\sum_{t=1}^{T}\sum_{q=1}^{Q}\|\mathbf{H}_{t,q}-\hat{\mathbf{H}}_{t,q}\|_{\mathrm{F}}/\sqrt{N_{\mathrm{ant}}N_{\mathrm{sub}}}\times100\%\), where \(\hat{\mathbf{H}}_{t,q}\) is obtained from the estimated radio map. The PADP-domain channel-distance deviation is \(E_{\mathrm{CD}}=\frac{1}{TQ}\sum_{t=1}^{T}\sum_{q=1}^{Q}D(\mathbf{H}_{t,q},\hat{\mathbf{H}}_{t,q})\), where \(D(\cdot,\cdot)\) is defined in \eqref{eq:padp_distance}.

We compare with six representative baselines. Siamese~\cite{StaYam:J23} and Bilateration~\cite{TanPal:J25} are channel-embedding-based methods; RITA~\cite{li2021wifi} and GraphIPS~\cite{zhao2020graphips} are IMU-assisted trajectory recovery methods; and ENAP~\cite{si2024environment} and VRLoc~\cite{si2025unsupervised} are unsupervised or weakly supervised positioning methods based on crowdsourced trajectory information. All baselines are adapted to the same measurement region and evaluated using the same ground-truth trajectory for testing. Any additional side information required by a baseline, such as IMU traces or post-hoc geometric alignment, is specified in its implementation.

For the proposed method, the sampling interval is \(\delta=0.2\,\mathrm{s}\), and the spatial graph is generated from the valid walkable area using the known deployment geometry. The regularization weight \(\eta\) is selected from a candidate set using an annotation-free criterion based on the normalized objective in \eqref{eq:P0_revised} and the PADP residual consistency in \eqref{eq:padp_pairwise_energy}; unless otherwise specified, \(\eta=3000\). Ground-truth positions are used only for metric computation and for plotting the sensitivity curve in Fig.~\ref{fig:errorVSeta}(c).

\begin{table}[t]
	\caption{Radio map construction and trajectory recovery performance of the proposed method compared with representative baselines.}
	\label{tab:rm-per}
	\centering
	\begin{tabular}{l|cccc}
		\toprule
		Method 
		& \(E_{\mathrm{loc}}\) (m)
		& \(E_{\mathrm{Beam}}\) (\%)
		& \(E_{\mathrm{RMSE}}\) (\%)
		& \(E_{\mathrm{CD}}\) \\
		\midrule
		RITA~\cite{li2021wifi}       & 7.58 & 22.85 & 24.09 & 16.49 \\
		GraphIPS~\cite{zhao2020graphips} & 2.11 & 8.68  & 10.76 & 8.51 \\
		ENAP~\cite{si2024environment}    & 3.86 & 11.61 & 10.83 & 9.81 \\
		VRLoc~\cite{si2025unsupervised}  & 2.24 & 9.62  & 8.42  & 7.53 \\
		Siamese~\cite{StaYam:J23}        & 1.21 & 7.08  & 7.17  & 3.14 \\
		Bilateration~\cite{TanPal:J25}   & 0.98 & 6.81  & 7.89  & 2.95 \\
		Proposed (w/o)              & 1.32 & 7.64  & 7.96  & 3.47 \\
		Proposed                         & \textbf{0.88} & \textbf{6.68} & \textbf{6.93} & \textbf{2.86} \\
		\bottomrule
	\end{tabular}
	\vspace{-0.15in}
\end{table}

\subsection{Trajectory Inference Performance}
\label{subsec:MapLocConst}

Table~\ref{tab:rm-per} reports the average localization error of all compared methods. The proposed method achieves an average localization error of \(0.88\,\mathrm{m}\), which is lower than the errors obtained by the representative baselines in this experiment. Compared with the unregularized variant, denoted as \emph{Proposed (w/o)}, the PADP-regularized method reduces the localization error from \(1.32\,\mathrm{m}\) to \(0.88\,\mathrm{m}\). This comparison indicates that the PADP continuity term provides useful geometric information beyond the RSS, directional, and mobility likelihoods.

The IMU-assisted methods, including RITA~\cite{li2021wifi} and GraphIPS~\cite{zhao2020graphips}, exhibit larger errors in this industrial dataset. One possible reason is that trajectory recovery based on inertial information is sensitive to drift accumulation and domain mismatch when the original assumptions of the methods are not fully satisfied. The clustering-based methods, including ENAP~\cite{si2024environment} and VRLoc~\cite{si2025unsupervised}, also show degraded performance, suggesting that discrete cluster-level structures alone may not be sufficient to preserve fine-grained trajectory continuity in this environment.

The strongest competing methods are Siamese~\cite{StaYam:J23} and Bilateration~\cite{TanPal:J25}, which achieve \(1.21\,\mathrm{m}\) and \(0.98\,\mathrm{m}\), respectively. The margin between Bilateration and the proposed method is relatively small, but the proposed method has the advantage of directly incorporating a CSI-derived spatial continuity prior during trajectory decoding. In addition, Bilateration and channel-charting-based methods often require geometric alignment or anchor information to map the learned chart to physical coordinates, whereas the proposed method uses the known AP deployment geometry and walkable domain to infer the trajectory without user-location labels.

\begin{figure}[t]
	\centering
	\includegraphics[width=1\columnwidth]{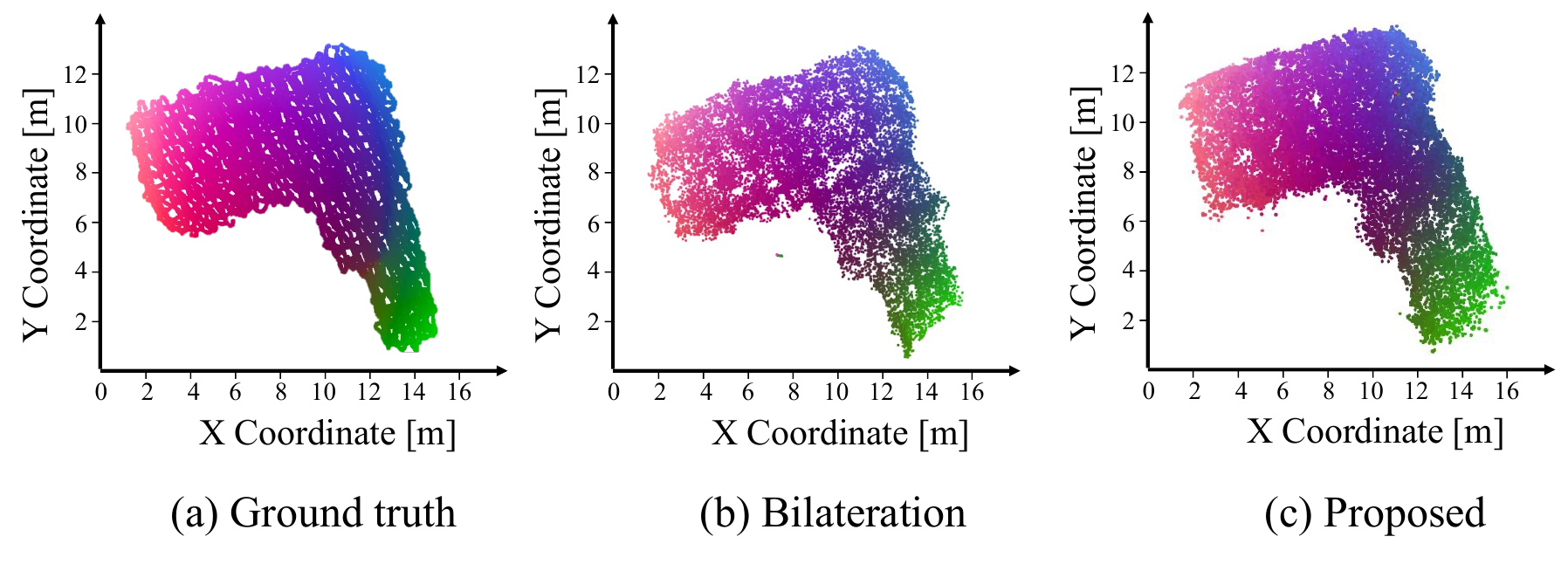}
	\vspace{-0.3in}
	\caption{Visual comparison of trajectory recovery: (a) ground-truth trajectory, (b) Bilateration~\cite{TanPal:J25}, and (c) the proposed method.}
	\label{fig:Ablation}
	\vspace{-0.2in}
\end{figure}

Fig.~\ref{fig:Ablation} provides a visual comparison among the ground-truth trajectory, Bilateration~\cite{TanPal:J25}, and the proposed method. The recovered trajectory of the proposed method better preserves the global shape of the measured path in this example. This visual result is consistent with the quantitative improvement in \(E_{\mathrm{loc}}\) reported in Table~\ref{tab:rm-per}. The comparison also illustrates the role of the PADP-based pairwise regularizer: it discourages transitions whose physical displacement is inconsistent with the observed CSI variation, thereby reducing discontinuities in the inferred trajectory.

\begin{figure}[t]
	\centering
	\includegraphics[width=1\columnwidth]{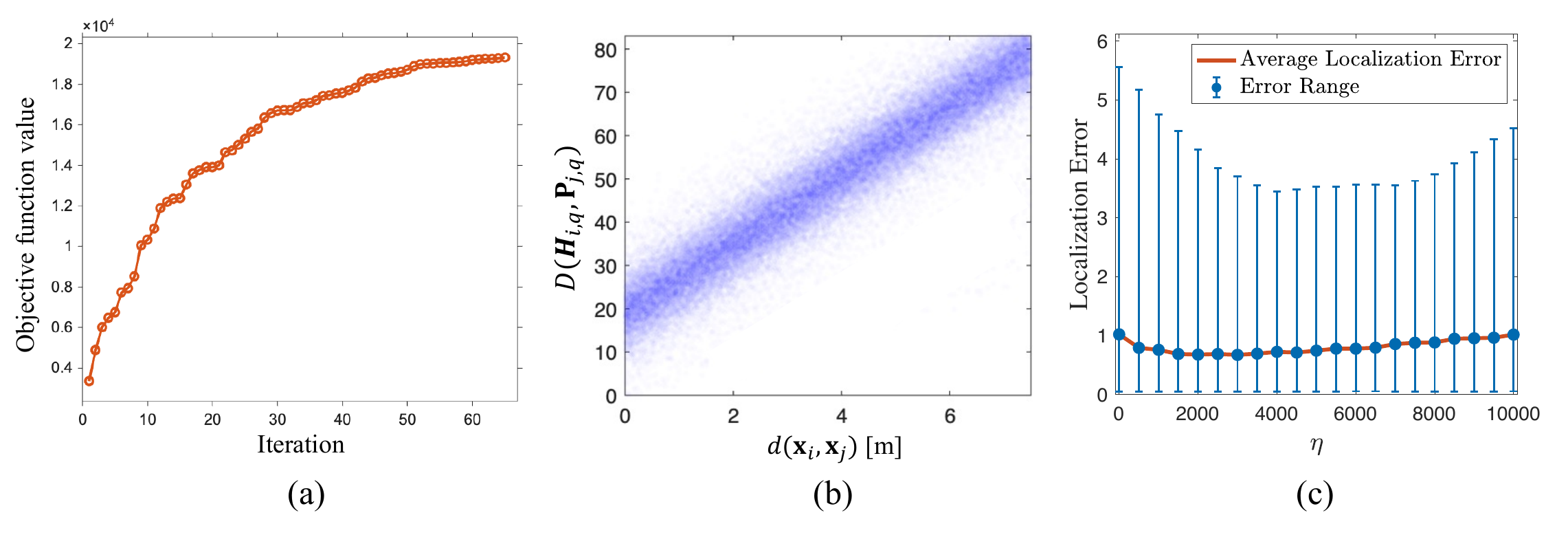}
	\vspace{-0.2in}
	\caption{Additional analyses of the proposed method: (a) convergence behavior of the alternating optimization algorithm, (b) empirical relationship between the PADP feature distance and Euclidean physical distance, and (c) influence of the regularization weight \(\eta\) on the localization error \(E_{\mathrm{loc}}\).}
	\label{fig:errorVSeta}
	\vspace{-0.2in}
\end{figure}

\subsection{Radio Map Construction Performance}
\label{subsec:ConsPer}

Table~\ref{tab:rm-per} also reports the radio map construction metrics. The proposed method obtains a relative beam estimation error of \(6.68\%\), a channel reconstruction RMSE of \(6.93\%\), and a channel-distance deviation of \(2.86\). These results indicate that the recovered trajectory can serve as a useful spatial anchor for constructing beam-domain and channel-domain radio maps.

Compared with the unregularized variant, the proposed method reduces \(E_{\mathrm{Beam}}\) from \(7.64\%\) to \(6.68\%\), \(E_{\mathrm{RMSE}}\) from \(7.96\%\) to \(6.93\%\), and \(E_{\mathrm{CD}}\) from \(3.47\) to \(2.86\). This suggests that improving trajectory consistency also benefits radio map reconstruction. In particular, inaccurate trajectory recovery may associate CSI observations with incorrect spatial locations, leading to blurred or topologically distorted radio maps. The PADP regularizer helps alleviate this issue by encouraging local consistency between CSI variation and physical displacement.

Fig.~\ref{fig:errorVSeta} provides further analysis of the proposed method. Fig.~\ref{fig:errorVSeta}(a) shows that the objective in \eqref{eq:P0_revised} increases rapidly and then stabilizes, indicating empirical convergence on the selected discrete or pruned state space. Fig.~\ref{fig:errorVSeta}(b) shows a positive correlation between PADP distance and Euclidean distance, supporting the local continuity motivation in Proposition~\ref{prop:local-continuity}, although the relation is not perfectly linear due to multipath, NLoS effects, finite angular-delay resolution, and residual synchronization errors. Fig.~\ref{fig:errorVSeta}(c) shows that the localization error decreases as \(\eta\) increases and reaches its best value around \(\eta=3000\), while excessive regularization slightly degrades performance.

\section{Conclusion}
\label{sec:Conclusion}

This paper proposed a location-label-free indoor radio mapping framework under known AP geometry and walkable-domain information. By integrating RSS, directional cues, mobility constraints, and a PADP-based spatial continuity regularizer, the proposed Bayesian inference framework recovers the latent trajectory from MIMO-OFDM CSI without user-location labels or IMU measurements. The recovered trajectory is then used to construct RSS, PADP, channel, and beam-domain radio maps. Experiments on a real-world industrial CSI dataset achieve \(0.88\,\mathrm{m}\) localization error and \(6.68\%\) relative beam map error, validating the effectiveness of the proposed PADP regularization.

\bibliographystyle{IEEEtran}
\bibliography{Source/IEEEabrv,Source/my_ref,Source/StringDefinitions}

\end{document}

%% file: Source/JCgroup_acronym.tex
\newac{speb}{SPEB}{square position error bound}
\newac[plural=EFIMs,firstplural=Fisher information matrices (EFIMs)]{efim}{EFIM}{Fisher information matrix}
\newac{ne}{NE}{Nash equilibrium}
\newac{mse}{MSE}{mean squared error}
\newac{toa}{TOA}{time-of-arrival}
\newac{snr}{SNR}{signal-to-noise ratio}
\newac{lan}{LAN}{local area network}
\newac{psd}{PSD}{positive semidefinite}
\newac{pd}{PD}{positive definite}
\newac{wrt}{w.r.t.}{with respect to}
\newac{lhs}{L.H.S.}{left hand side}
\newac{wp1}{w.p.1}{with probability 1}
\newac{kkt}{KKT}{Karush-Kuhn-Tucker}
\newac{wlog}{w.l.o.g.}{without loss of generality}
\newac{mle}{MLE}{maximum likelihood estimation}
\newac{gps}{GPS}{global positioning system}
\newac{rssi}{RSSI}{received signal strength indicator}
\newac{mimo}{MIMO}{multiple-input multiple-output}
\newac{csi}{CSI}{channel state information}
\newac{fdd}{FDD}{frequency division duplexing}
\newac{ms}{MS}{mobile station}
\newac{bs}{BS}{base station}
\newac{d2d}{D2D}{device-to-device}
\newac{slnr}{SLNR}{signal-to-interference-leakage-and-noise-ratio}
\newac{ula}{ULA}{uniform linear antenna array}
\newac{pas}{PAS}{power angular spectrum}
\newac{mmse}{MMSE}{minimum mean square error}
\newac{zf}{ZF}{zero-forcing}
\newac{rzf}{RZF}{regularized zero-forcing}
\newac{as}{AS}{angular spread}
\newac{aod}{AOD}{angle of departure}
\newac{iid}{i.i.d.}{independent and identically distributed} 
\newac{sinr}{SINR}{signal-to-interference-and-noise ratio}
\newac{tdd}{TDD}{time-division duplex}
\newac{rvq}{RVQ}{random vector quantization}
\newac{rhs}{R.H.S.}{right hand side}
\newac{mrc}{MRC}{maximum ratio combining}
\newac{cdf}{CDF}{cumulative distribution function}
\newac{a.s.}{a.s.}{almost surely}
\newac{los}{LOS}{line-of-sight}
\newac{jsdm}{JSDM}{joint spatial division and multiplexing}
\newac{map}{MAP}{maximum a posteriori}
\newac{klt}{KLT}{Karhunen-Lo\`eve Transform}
\newac{lbe}{LBE}{link bargaining equilibrium}
\newac{se}{SE}{Stackelberg equilibrium}
\newac{uav}{UAV}{unmanned aerial vehicle}
\newac{nlos}{NLOS}{non-line-of-sight}
\newac{pdf}{PDF}{probability density function}
\newac{em}{EM}{expectation-maximization}
\newac{knn}{KNN}{$k$-nearest neighbor}
\newac{svd}{SVD}{singular value decomposition}
\newac{nmf}{NMF}{non-negative matrix factorization}
\newac{umf}{UMF}{unimodality-constrained matrix factorization}
\newac{rmse}{RMSE}{rooted mean squared error}
\newac{olos}{OLOS}{obstructed line-of-sight}
\newac{mmw}{mmW}{millimeter wave}
\newac{ber}{BER}{bit error rate}
\newac{rss}{RSS}{received signal strength}
\newac{lp}{LP}{linear program}
\newac{ufw}{U-FW}{unimodal Frank-Wolfe}
\newac{utf}{UTF}{unimodality-constrained tensor factorization}
\newac{fw}{FW}{Frank-Wolfe}
\newac{iot}{IoT}{Internet-of-Things}
\newac{mae}{MAE}{mean absolute error}
\newac{crb}{CRB}{Cram\'er-Rao bound}
\newac{aoa}{AoA}{angle of arrival}
\newac{wcl}{WCL}{weighted centroid localization}